\documentclass{article}

\usepackage{times}
\usepackage{graphicx}
\usepackage{color}
\usepackage{amssymb}
\usepackage{latexsym}
\usepackage{amsmath}
\usepackage{theorem}
\usepackage{csquotes}
\usepackage{verbatim}
\usepackage{tikz}
\usepackage{circuitikz}
\usetikzlibrary{arrows,shapes.gates.logic.US,shapes.gates.logic.IEC,calc}

\newtheorem{theorem}{Theorem}
\newtheorem{proposition}[theorem]{Proposition}

\newtheorem{corollary}[theorem]{Corollary}

\newtheorem{definition}[theorem]{Definition}
\newenvironment{proof}{\noindent\textbf{Proof:}\ }{\hfill$\Box$\medskip\par}
{\theorembodyfont{\rmfamily}}
{\theorembodyfont{\rmfamily}}

\title{A Note on the Power of Non-Deterministic Circuits with Gate Restrictions}
\author{Gustav Nordh\thanks{E-mail: {\tt gustav.nordh@gmail.com}}
}
\date{\today}

\begin{document}
\maketitle
\begin{abstract}
We investigate the power of non-deterministic circuits over restricted sets of base gates. We note that the power of non-deterministic circuits exhibit a dichotomy, in the following sense: For weak enough bases, non-deterministic circuits are no more powerful than deterministic circuits, and for the remaining bases, non-deterministic circuits are super polynomial more efficient than deterministic circuits (under the assumption that $P/poly \neq NP/poly$). Moreover, we give a precise characterization of the borderline between the two situations.
\end{abstract}

\section{Introduction}
In this note we are interested in finding evidence in favor of the belief that non-deterministic circuits are more powerful than deterministic circuits, and $P/poly \neq NP/poly$. Except for the well-known result that the polynomial hierarchy collapses if $P/poly = NP/poly$~\cite{KL80}, there seems to be little formal evidence in this direction. For example, we do not know of a function family $\{f_n\}_{n \geq 1}$ having non-deterministic circuit complexity $kn$, and deterministic circuit complexity of at least $(k+\epsilon)n$. 

A natural path for making progress on this question is to prove that non-deterministic circuits over some restricted set of base gates $G$ (i.e., non-deterministic $G$-circuits) are more powerful than deterministic $G$-circuits. Unfortunately, we are not aware of any results in this direction either. On the contrary, for one of the most well investigated restriction on $G$, namely monotone circuits (i.e., $G =\{\land,\lor,0,1\}$), it is well known that non-deterministic $G$-circuits are no more powerful than deterministic $G$-circuits.

Given a finite set of base gates $G$, let $[G]$ denote the set of all functions/gates that can be implemented by circuits over $G$ (i.e., implemented by a $G$-circuit). Sets of the form $[G]$ are called clones, and Post~\cite{post} classified all $[G]$ over the Boolean domain. Sets $[G]$ form a lattice under set inclusion and the cardinality of the lattice is countable infinite. For example, the class of all monotone Boolean functions can be computed by circuits consisting of only $\land$ and $\lor$ gates and constants $0$ and $1$, or in other words, the clone of all monotone Boolean functions is $[\land, \lor, 0, 1]$. For more information on Post's lattice, see~\cite{signew1}.

By making heavy use of Post's lattice, we are able to make the following observations: If the constants $0,1$ are in the base $G$, and $G$ is not a full base (i.e., $[G] \neq [\land,\lor,\neg]$), then non-deterministic $G$-circuits are no more powerful than deterministic $G$-circuits. Furthermore, if the base $G$ is monotone, linear, or self-dual, then non-deterministic $G$-circuits are no more powerful than deterministic $G$-circuits. For the remaining bases $G$, non-deterministic $G$-circuits are super polynomial more efficient than deterministic $G$-circuits (under the assumption that $P/poly \neq NP/poly$). 

\subsection{Preliminaries}
A Boolean circuit is a directed acyclic graph with three types of labeled vertices: sources (in-degree $0$) labeled $x_1,\dots,x_n$, a sink (the output), and vertices with in-degree $k > 0$ are gates labeled by Boolean functions on $k$ inputs.
A non-deterministic circuit has, in addition to the inputs $x = (x_1,\dots,x_n)$, a set of \enquote{non-deterministic} inputs $y=(y_1,\dots,y_m)$. A non-deterministic circuit $C$ accepts input $x$ if there exists $y$ such that the circuit output $1$ on $(x,y)$.  
Let $|C|$ denote the number of gates of a circuit $C$.

A family of non-deterministic circuits $\{C_n\}_{n \geq 1}$, with $C_n$ having $n$ (ordinary) inputs, decide a language $L$ if each $C_n$ decide $L_n$ (i.e., $C_n$ accepts $x$ if and only if $|x| =n$ and $x \in L$). 
The class $NP/poly$ is defined as the class of languages decidable by non-deterministic circuit families $\{C_n\}$, with $|C_n| \leq poly(n)$. Recall that $P/poly$ is the class of languages decidable by (deterministic) circuit families $\{C_n\}$, with $|C_n| \leq poly(n)$.

By a circuit we always mean a deterministic circuit, unless it is explicitly said to be a non-deterministic circuit, or it is clear from the context. We assume all sets of base gates $G$ to be finite.
By abusing notation slightly, if a language $L$ is decidable by deterministic $G$-circuits, we denote this by $L \in [G]$.
\begin{definition}
We say that a set of base gates $G$ lack non-deterministic power, if any language $L \in [G]$ that has non-deterministic $G$-circuit complexity $s(n)$, has deterministic $G$-circuit complexity $O(s(n))$.
\end{definition}

\begin{definition}
We say that a set of base gates $G$ has full non-deterministic power if, under the assumption that $P/poly \neq NP/poly$, there is a language $L \in [G]$ that has polynomial non-deterministic $G$-circuit complexity, but does not have polynomial deterministic $G$-circuit complexity.
\end{definition}

\section{Classification}
\begin{proposition}
Given finite sets of base gates $G_1$ and $G_2$ with $[G_1] = [G_2]$, then $G_1$ lack non-deterministic power if and only if $G_2$ lack non-deterministic power.
\label{prop:clone}
\end{proposition}
\begin{proof}
If $G_1 \subseteq [G_2]$, then every gate $g(x_1,\dots,x_k) \in G_1$ has an implementation of size $c^k$ using gates from $G_2$. Hence, we can convert any $G_1$-circuit into an equivalent $G_2$-circuit without blowing up the size more than a constant factor. Similarly, if $G_2 \subseteq [G_1]$, then any $G_2$-circuit can be converted to an equivalent $G_1$-circuit without increasing the size more than a constant factor. The same holds for non-deterministic circuits, and the result follows.
\end{proof}

\begin{corollary}
Given finite sets of base gates $G_1$ and $G_2$ with $[G_1] = [G_2]$, then $G_1$ has full non-deterministic power if and only if $G_2$ has full non-deterministic power.
\end{corollary}
\begin{proof}
By the proof of Proposition~\ref{prop:clone}.
\end{proof}

\begin{definition}
A Boolean function $f(x_1,\dots,x_n)$ is said to be self-dual if $f(x_1,\dots,x_n) = \overline{f(\overline{x_1},\dots,\overline{x_n})}$ for all $x_1,\dots,x_n \in \{0,1\}$, where $\overline{0}=1$ and $\overline{1}=0$.
\end{definition}
From Post's classification of Boolean clones~\cite{post} we know that the function $d(x_1,x_2,x_3) = (x_1 \land \neg x_2) \lor (\neg x_2 \land \neg x_3) \lor (\neg x_3 \land x_1)$ generates the clone of all self-dual Boolean functions. In other words, any self-dual Boolean function can be computed by a circuit consisting of $d(x_1,x_2,x_3)$-gates (from now on referred to as $d$-gates).
\begin{definition}
A set of base gates $G$ is said to be
\begin{itemize}
\item monotone if $G \subseteq [\land,\lor, 0, 1]$,
\item linear if $G \subseteq [\oplus, 1]$, and
\item self-dual if $G \subseteq [d(x_1,x_2,x_3)]$.
\end{itemize}
\end{definition}

\begin{proposition}
If $G$ is self-dual, then $G$ lacks non-deterministic power.
\end{proposition}
\begin{proof}
Given a non-deterministic $d$-circuit $C(x,y)$ with $x=(x_1,\dots,x_n)$ and $y=(y_1\dots,y_m)$ computing $f$, recall that $f(x) = 1$ if and only if there is a $y$ such that $C(x,y) = 1$. Assume there is $y$ and $y'$ such that $C(x,y) = 1$ and $C(x,y') = 0$, then $C(\overline{x},\overline{y'}) = 1$, because $C$ consists of $d$-gates which are self-dual. Hence, $f(x) = 1 = f(\overline{x})$ which is impossible since $f$ is self-dual. Thus, if $f(x) = 1$, then $C(x,y) = 1$ for all $y$. To construct the equivalent deterministic $d$-circuit $C'(x)$ we replace all $y_i$ inputs in $C(x,y)$ with $x_1$ (i.e., we replace each non-deterministic variables by the ordinary variable $x_1$). Note that the more natural transformation of replacing the $y_i$'s by constants, does not work, since the resulting circuit is then not necessarily a $d$-circuit.
\end{proof}

\begin{proposition}
If $G$ is monotone, then $G$ lacks non-deterministic power.
\end{proposition}
\begin{proof}
The result follows from the fact that a non-deterministic monotone $G$-circuit $C(x,y)$ outputs $1$ on  input $x,y$ if and only if $C(x,1) = 1$. That is, given a non-deterministic monotone $G$-circuit $C(x,y)$ we can construct an equivalent deterministic monotone $G$-circuit, without increasing the size, by replacing all $y$ variables with the constant $1$.
\end{proof}

\begin{proposition}
If $G$ is linear, then $G$ lacks non-deterministic power.
\end{proposition}
\begin{proof}
Any function $L \in [G]$ for a linear base $G$ can be computed by a deterministic $G$-circuit of size $O(n)$.
\end{proof}

\begin{proposition}
Any $G$ such that $\{x \land (y \lor \neg z)\} \subseteq [G]$ has full non-deterministic power. 
\end{proposition}
\begin{proof}
By Post's lattice, $[G,0,1] = [\land,\lor,\neg]$.
Hence, any circuit (over any finite basis) can be converted into an equivalent $\{G,0,1\}$ circuit $C$, without blowing up the size more than a constant factor. Note that $x \land (y \lor z) \in [G]$, and consider the $G$-circuit $C' = x' \land (C \lor x'')$, where all $1$'s and $0$'s in $C$ have been replaced by $x'$ and $x''$ respectively. The transformation can be carried out both for deterministic and non-deterministic circuits.

Given $L$, having polynomial non-deterministic complexity and super polynomial deterministic complexity over the full basis (such an $L$ exist under the assumption that $P/poly \neq NP/poly$). Consider, $L' = \{(x,1,0) \mid x \in L\} \cup \{(x,1,1) \mid x \in \{0,1\}^*\}$. If $C_n(x,y)$ is a family of non-deterministic circuits deciding $L$, then $C_n'(x,x',x'',y)$ (as defined above) is a family of non-deterministic circuits over the basis $[G]$ deciding $L'$. Hence, $L'$ has polynomial non-deterministic complexity over the basis $[G]$. Assume towards contradiction that $L'$ has polynomial deterministic complexity over the basis $[G]$, i.e., there is a polynomial size circuit family $C'_n(x,x',x'')$ deciding $L'$. Then, $C'_n(x,1,0)$ is a polynomial size circuit family over the full basis $[\land,\lor,\neg]$ deciding $L$, contradicting that $L$ has super polynomial deterministic circuit complexity.
\end{proof}

\begin{proposition}
Any $G$ such that $\{x \lor (y \land \neg z)\} \subseteq [G]$ has full non-deterministic power.
\end{proposition}
\begin{proof}
By Post's lattice, $[G,0,1] = [\land,\lor,\neg]$.
Hence, any circuit (over any finite basis) can be converted into an equivalent $\{G,0,1\}$ circuit $C$, without blowing up the size more than a constant factor. Note that $x \lor (y \land z) \in [G]$, and consider the $G$-circuit $C' = x' \lor (C \land x'')$, where all $0$'s and $1$'s in $C$ have been replaced by $x'$ and $x''$ respectively. The transformation can be carried out both for deterministic and non-deterministic circuits.

Given $L$, having polynomial non-deterministic complexity and super polynomial deterministic complexity over the full basis.
Consider, $L' = \{(x,0,1) \mid x \in L\} \cup \{(x,1,x'') \mid x \in \{0,1\}^* \; and \; x'' \in \{0,1\}\}$. If $C_n(x,y)$ is a family of non-deterministic circuits deciding $L$, then $C_n'(x,x',x'',y)$ is a family of non-deterministic circuits over the basis $[G]$ deciding $L'$. Hence, $L'$ has polynomial non-deterministic complexity over the basis $[G]$. Assume towards contradiction that $L'$ has polynomial deterministic complexity over the basis $[G]$, i.e., there is a polynomial size circuit family $C'_n(x,x',x'')$ deciding $L'$. Then, $C'_n(x,0,1)$ is a polynomial size circuit family over the full basis $[\land,\lor,\neg]$ deciding $L$, contradicting that $L$ has super polynomial deterministic circuit complexity.
\end{proof}

\begin{proposition}
If $G$ is self-dual, monotone, or linear, then $G$ lack non-deterministic power. All other $G$ have full non-deterministic power.
\end{proposition}
\begin{proof}
By the results above and inspection of Post's lattice.
In particular, if $G$ is neither self-dual, monotone, nor linear, then by Post's lattice $\{x \land (y \lor \neg z)\} \subseteq [G]$ or $\{x \lor (y \land \neg z)\} \subseteq [G]$. 
\end{proof}

\begin{corollary}
Any $G$ such that $\{0,1\} \subseteq [G]$ and $[G] \neq [\land,\lor,\neg]$, lacks non-deterministic power.
\end{corollary}
\begin{proof}
By the results above and inspection of Post's lattice.
\end{proof}

\section{Final Remarks}
Coming back to our original motivation for studying non-deterministic versus deterministic $G$-circuit complexity. We note that $G = \{x \land (y \lor \neg z)\}$ and $G = \{x \lor (y \land \neg z)\}$ are the two weakest bases for which it is possible that non-deterministic $G$-circuits are more powerful than deterministic $G$-circuits (indeed, this is the case assuming $P/poly \neq NP/poly$). Unfortunately, as the following two propositions show, it is probably somewhat challenging to prove super polynomial lower bounds for deterministic $G$-circuits (with $G = \{x \land (y \lor \neg z)\}$ or $G = \{x \lor (y \land \neg z)\}$), since this implies that $P/poly \neq NP/poly$.

\begin{proposition}
Let $L \in G = [x \land (y \lor \neg z)]$ with deterministic circuit complexity $s(n)$ over the full basis $[\land, \lor, \neg]$. Then, $L$ has $O(s(n))$ deterministic $G$-circuit complexity.
\end{proposition}
\begin{proof}
First note that $G$ is $1$-reproducing, which means that all constant $1$ vectors are in $L$ (i.e., $(1,\dots,1) \in L$). We first show that such an $L$ can be computed by $\{\land, \lor, x \oplus y \oplus 1\}$ circuits of size $O(s(n))$ if it can be computed by $[\land, \lor, \neg]$ circuits of size $s(n)$. Hence, it suffices to show how to simulate $\neg$ gates with $x \oplus y \oplus 1$ gates. Given a $\{\land, \lor,\neg\}$ circuit for $L_n$ with inputs $x_1,\dots, x_n$ we construct an equivalent circuit $C_n$ by replacing all $\neg$ gates with $x \oplus y \oplus 1$ gates, where the first input is the original input to the $\neg$ gate and the second input is the output of $x_1 \land \dots \land x_n$.

Observe that since $\{\land, \lor, x \oplus y \oplus 1\}$ is $1$ reproducing, $C_n(1,\dots,1) = 1$ (as it should, since $(1,\dots,1) \in L_n$). For every $(x_1,\dots,x_n) \neq (1, \dots, 1)$, the $x \oplus y \oplus 1$ gates works like a $\neg$ gates on its first input (since $x_1 \land \dots \land x_n = 0$). 

By Post's lattice we know that $[\land, \lor, x \oplus y \oplus 1] = [G,1]$, and hence $L$ has circuit complexity $O(s(n))$ over $\{G,1\}$, via a circuit family $C_n$.
The base $G = [x \land (y \lor \neg z)]$ is $1$-separating, which means that for any $G$-circuit $C(x_1,\dots,x_n)$ there is an $1 \leq i \leq n$, such that $x_i = 1$ for all $\{(x_1,\dots,x_n) \mid C(x_1,\dots,x_n) = 1\}$.
Since $L \in G$ and $G$ is $1$-separating, there is for each $L_n$, an $1 \leq i \leq n$ such that $x_i = 1$ for all $(x_1,\dots,x_n) \in L_n$.
Hence, $C'_n = x_i \land C_n$ (replacing all occurrences of $1$ in $C_n$ by $x_i$) is a $G$-circuit family of size $O(s(n))$ deciding the language $L$. 
\end{proof}

\begin{proposition}
Let $L \in G = [x \lor (y \land \neg z)]$ with deterministic circuit complexity $s(n)$ over the full basis $[\land, \lor, \neg]$. Then, $L$ has $O(s(n))$ deterministic $G$-circuit complexity.
\end{proposition}
\begin{proof}
First note that $G$ is $0$-reproducing, which means that all constant $0$ vectors are in $L$ (i.e., $(0,\dots,0) \in L$). We first show that such an $L$ can be computed by $\{\land, \lor,  \oplus \}$ circuits of size $O(s(n))$ if it can be computed by $[\land, \lor, \neg]$ circuits of size $s(n)$. Hence, it suffices to show how to simulate $\neg$ gates with $ \oplus $ gates. Given a $\{\land,\lor,\neg\}$ circuit for $L_n$ with inputs $x_1,\dots, x_n$ we construct an equivalent circuit $C_n$ by replacing all $\neg$ gates with $\oplus $ gates, where the first input is the original input to the $\neg$ gate and the second input is the output of $x_1 \lor \dots \lor x_n$.

Observe that since $\{\land, \lor, \oplus \}$ is $0$ reproducing, $C_n(0,\dots,0) = 0$ (as it should, since $(0,\dots,0) \in L_n$). For every $(x_1,\dots,x_n) \neq (0, \dots, 0)$, the $\oplus$ gates works like a $\neg$ gates on its first input (since $x_1 \lor \dots \lor x_n = 1$). 

By Post's lattice we know that $[\land, \lor, \oplus ] = [G,0]$, and hence $L$ has circuit complexity $O(s(n))$ over $\{G,0\}$, via a circuit family $C_n$.
The base $G = [x \lor (y \land \neg z)]$ is $0$-separating, which means that for any $G$-circuit $C(x_1,\dots,x_n)$ there is an $1 \leq i \leq n$, such that $x_i = 0$ for all $\{(x_1,\dots,x_n) \mid C(x_1,\dots,x_n) = 0\}$.
Since $L \in G$ and $G$ is $0$-separating, there is for each $L_n$, an $1 \leq i \leq n$ such that $x_i = 0$ for all $(x_1,\dots,x_n) \notin L_n$.
Hence, $C'_n = x_i \lor C_n$ (replacing all occurrences of $0$ in $C_n$ by $x_i$) is a $G$-circuit family of size $O(s(n))$ deciding the language $L$. 
\end{proof}

\bibliographystyle{abbrv}
\bibliography{references}

\end{document}